\documentclass[11pt]{article}

\usepackage{times}

\usepackage[paperwidth=199.8mm,paperheight=297mm,centering,hmargin=25mm,vmargin=25mm]{geometry}
\usepackage{authblk} 
\usepackage[titletoc,title]{appendix}
\usepackage{hyperref}
\hypersetup{colorlinks=true,citecolor=blue,linkcolor=blue,filecolor=blue,urlcolor=blue,breaklinks=true}
\usepackage[dvipsnames]{xcolor}
\usepackage{color}
\usepackage[marginal]{footmisc}
\usepackage{url}

\usepackage{mathtools}
\usepackage{amsmath}
\usepackage{amssymb}
\usepackage[shortlabels]{enumitem}
\usepackage{graphicx,epic,eepic,epsfig,latexsym,verbatim}
\usepackage{dsfont}

\usepackage{framed}
\definecolor{shadecolor}{rgb}{0.9,0.9,0.9}
\usepackage{subcaption}
\usepackage{caption}
\usepackage{float}
\usepackage{tikz}
\usepackage{tcolorbox}
\usepackage{relsize}

\definecolor{LightGray}{RGB}{220,220,220}
\definecolor{myred}{RGB}{236, 17, 0}
\definecolor{myblue}{RGB}{10, 88, 153}
\definecolor{myorange}{RGB}{236, 137, 0}
\definecolor{mygreen}{RGB}{26, 152, 81}

\usepackage{amsthm}
\newtheorem{definition}{Definition}
\newtheorem{proposition}[definition]{Proposition}

\newtheorem{theorem}[definition]{Theorem}
\newtheorem{corollary}[definition]{Corollary}

\def\squareforqed{\hbox{\rlap{$\sqcap$}$\sqcup$}}
\def\qed{\ifmmode\squareforqed\else{\unskip\nobreak\hfil
\penalty50\hskip1em\null\nobreak\hfil\squareforqed
\parfillskip=0pt\finalhyphendemerits=0\endgraf}\fi}
\def\endenv{\ifmmode\;\else{\unskip\nobreak\hfil
\penalty50\hskip1em\null\nobreak\hfil\;
\parfillskip=0pt\finalhyphendemerits=0\endgraf}\fi}

\newcounter{remark}
\newenvironment{remark}[1][]{\refstepcounter{remark}\par\medskip\noindent%
\textbf{Remark~\theremark #1} }{\medskip}

\newcounter{example}

\mathchardef\ordinarycolon\mathcode`\:
\mathcode`\:=\string"8000
\def\vcentcolon{\mathrel{\mathop\ordinarycolon}}
\begingroup \catcode`\:=\active
  \lowercase{\endgroup
  \let :\vcentcolon
  }

\usepackage{colortbl}
\usepackage{pifont}
\definecolor{Gray}{gray}{0.92}
\definecolor{Gray2}{gray}{0.75}
\definecolor{maroon}{cmyk}{0,0.87,0.68,0.32}
\usepackage{booktabs}
\usepackage{makecell}
\usepackage{diagbox}
\usepackage{multirow}

\newcommand{\nc}{\newcommand}
\nc{\rnc}{\renewcommand}
\nc{\bra}[1]{\langle#1|}
\nc{\ket}[1]{|#1\rangle}
\nc{\<}{\langle}
\rnc{\>}{\rangle}
\nc{\ketbra}[2]{|#1\rangle\!\langle#2|}
\nc{\braket}[2]{\langle#1|#2\rangle}
\nc{\braandket}[3]{\langle #1|#2|#3\rangle}
\nc{\proj}[1]{| #1\rangle\!\langle #1 |}
\nc{\avg}[1]{\langle#1\rangle}
\nc{\Rank}{\operatorname{Rank}}
\nc{\smfrac}[2]{\mbox{$\frac{#1}{#2}$}}
\nc{\tr}{\operatorname{Tr}}
\nc{\ox}{\otimes}

\nc{\cA}{{\cal A}}
\nc{\cB}{{\cal B}}
\nc{\cC}{{\cal C}}
\nc{\cD}{{\cal D}}
\nc{\cE}{{\cal E}}
\nc{\cF}{{\cal F}}
\nc{\cG}{{\cal G}}
\nc{\cH}{{\cal H}}
\nc{\cI}{{\cal I}}
\nc{\cJ}{{\cal J}}
\nc{\cK}{{\cal K}}
\nc{\cL}{{\cal L}}
\nc{\cM}{{\cal M}}
\nc{\cN}{{\cal N}}
\nc{\cO}{{\cal O}}
\nc{\cP}{{\cal P}}
\nc{\cQ}{{\cal Q}}
\nc{\cR}{{\cal R}}
\nc{\cS}{{\cal S}}
\nc{\cT}{{\cal T}}
\nc{\cV}{{\cal V}}
\nc{\cU}{{\cal U}}
\nc{\cX}{{\cal X}}
\nc{\cY}{{\cal Y}}
\nc{\cZ}{{\cal Z}}
\nc{\cW}{{\cal W}}

\nc{\RR}{{{\mathbb R}}}
\nc{\CC}{{{\mathbb C}}}
\nc{\FF}{{{\mathbb F}}}
\nc{\NN}{{{\mathbb N}}}
\nc{\ZZ}{{{\mathbb Z}}}
\nc{\QQ}{{{\mathbb Q}}}
\nc{\UU}{{{\mathbb U}}}
\nc{\EE}{{{\mathbb E}}}
\nc{\id}{{\operatorname{id}}}

\nc{\1}{{\mathds{1}}}
\nc{\supp}{{\operatorname{supp}}}

\makeatletter
\def\grd@save@target#1{%
  \def\grd@target{#1}}
\def\grd@save@start#1{%
  \def\grd@start{#1}}
\tikzset{
  grid with coordinates/.style={
    to path={%
      \pgfextra{%
        \edef\grd@@target{(\tikztotarget)}%
        \tikz@scan@one@point\grd@save@target\grd@@target\relax
        \edef\grd@@start{(\tikztostart)}%
        \tikz@scan@one@point\grd@save@start\grd@@start\relax
        \draw[minor help lines,magenta] (\tikztostart) grid (\tikztotarget);
        \draw[major help lines] (\tikztostart) grid (\tikztotarget);
        \grd@start
        \pgfmathsetmacro{\grd@xa}{\the\pgf@x/1cm}
        \pgfmathsetmacro{\grd@ya}{\the\pgf@y/1cm}
        \grd@target
        \pgfmathsetmacro{\grd@xb}{\the\pgf@x/1cm}
        \pgfmathsetmacro{\grd@yb}{\the\pgf@y/1cm}
        \pgfmathsetmacro{\grd@xc}{\grd@xa + \pgfkeysvalueof{/tikz/grid with coordinates/major step}}
        \pgfmathsetmacro{\grd@yc}{\grd@ya + \pgfkeysvalueof{/tikz/grid with coordinates/major step}}
        \foreach \x in {\grd@xa,\grd@xc,...,\grd@xb}
        \node[anchor=north] at (\x,\grd@ya) {\pgfmathprintnumber{\x}};
        \foreach \y in {\grd@ya,\grd@yc,...,\grd@yb}
        \node[anchor=east] at (\grd@xa,\y) {\pgfmathprintnumber{\y}};
      }
    }
  },
  minor help lines/.style={
    help lines,
    step=\pgfkeysvalueof{/tikz/grid with coordinates/minor step}
  },
  major help lines/.style={
    help lines,
    line width=\pgfkeysvalueof{/tikz/grid with coordinates/major line width},
    step=\pgfkeysvalueof{/tikz/grid with coordinates/major step}
  },
  grid with coordinates/.cd,
  minor step/.initial=.2,
  major step/.initial=1,
  major line width/.initial=2pt,
}
\makeatother

\makeatletter
\newcommand*\rel@kern[1]{\kern#1\dimexpr\macc@kerna}
\newcommand*\widebar[1]{%
  \begingroup
  \def\mathaccent##1##2{%
    \rel@kern{0.8}%
    \overline{\rel@kern{-0.8}\macc@nucleus\rel@kern{0.2}}%
    \rel@kern{-0.2}%
  }%
  \macc@depth\@ne
  \let\math@bgroup\@empty \let\math@egroup\macc@set@skewchar
  \mathsurround\z@ \frozen@everymath{\mathgroup\macc@group\relax}%
  \macc@set@skewchar\relax
  \let\mathaccentV\macc@nested@a
  \macc@nested@a\relax111{#1}%
  \endgroup
}
\makeatother

\newcommand{\reg}{\infty}

\newcommand{\Renyi}{R\'{e}nyi }

\usepackage{mathrsfs}
\nc{\pl}{{\scalebox{0.7}{+}}}
\nc{\PSD}{\HERM_{\pl}}
\nc{\PD}{\HERM_{\pl\pl}}
\nc{\polarPSD}[1]{{#1}_{\pl}^{\circ}}
\nc{\polarPSDre}[1]{{#1}_{\pl}^{\star}}
\nc{\polarPD}[1]{{#1}_{\pl\pl}^{\circ}}
\nc{\HERM}{\mathscr{H}}
\nc{\cvxset}{\mathscr{C}}
\nc{\density}{\mathscr{D}}
\nc{\subdensity}{\mathscr{D}_\bullet}
\nc{\bcC}{\cC}
\nc{\Meas}{{\scriptscriptstyle \rm M}}
\nc{\Proj}{{{\scriptscriptstyle \rm P}}}
\nc{\RM}{{{\mathscr{R}}}}
\nc{\sK}{{{\mathscr{K}}}}
\nc{\sS}{{{\mathscr{S}}}}
\nc{\sT}{{{\mathscr{T}}}}
\nc{\sA}{{{\mathscr{A}}}}
\nc{\sB}{{{\mathscr{B}}}}
\nc{\sC}{{{\mathscr{C}}}}
\nc{\sE}{{{\mathscr{E}}}}
\nc{\sL}{{{\mathscr{L}}}}
\nc{\sG}{{{\mathscr{G}}}}
\nc{\sF}{{{\mathscr{F}}}}
\nc{\sI}{{{\mathscr{I}}}}
\nc{\sN}{{{\mathscr{N}}}}
\nc{\sM}{{{\mathscr{M}}}}
\nc{\END}{\operatorname{End}}
\nc{\PERM}{\mathfrak{\sigma}}

\nc{\Cone}{\text{\rm Cone}}
\nc{\sep}{{\SEP}}
\nc{\BS}{{\scriptscriptstyle \rm {BS}}}
\nc{\Sand}{{\scriptscriptstyle  \rm S}}
\nc{\Petz}{{\scriptscriptstyle  \rm P}}
\nc{\Hypo}{{\scriptscriptstyle  \rm H}}
\nc{\DD}{{{\mathbb D}}}
\nc{\suchthat}{\text{\rm s.t.}}
\nc{\PPT}{\text{\rm PPT}}
\nc{\Rains}{\text{\rm Rains}}
\nc{\WD}{\text{\rm WD}}
\nc{\new}{\text{\rm new}}
\nc{\sfT}{\mathsf T}
\nc{\SEP}{\text{\rm SEP}}
\nc{\PSEP}{\text{\rm PSEP}}
\nc{\CPTP}{\text{\rm CPTP}}
\nc{\POVM}{\text{\rm POVM}}
\nc{\PVM}{\text{\rm PVM}}
\nc{\CP}{\text{\rm CP}}
\nc{\adv}{\text{\rm adv}}
\nc{\spec}{\text{\rm spec}}
\nc{\poly}{\text{\rm poly}}
\nc{\End}{\operatorname{End}}
\nc{\Par}{\operatorname{Par}}
\nc{\RNG}{\operatorname{RNG}}
\nc{\epi}{\boldsymbol{\operatorname{epi}}}
\nc{\op}{\boldsymbol{\operatorname{op}}}

\usepackage{stmaryrd}
\nc{\db}[1]{\left\llbracket#1 \right\rrbracket}

\usepackage{pgfplots}

\nc{\img}{\mathbf{i}}

\usepackage{changepage}

\DeclareMathOperator{\sign}{sign}
\DeclareMathOperator{\dom}{dom}

\newcommand{\range}{\mathrm{range}}

\begin{document}

\title{\Large \textbf{{Uhlmann's theorem for measured divergences}}}

\author[1]{Kun Fang \thanks{kunfang@cuhk.edu.cn}}
\author[2]{Hamza Fawzi \thanks{h.fawzi@damtp.cam.ac.uk}}
\author[3]{Omar Fawzi \thanks{omar.fawzi@ens-lyon.fr}}

\affil[1]{\small School of Data Science, The Chinese University of Hong Kong, Shenzhen,\protect\\  Guangdong, 518172, China}
\affil[2]{\small Department of Applied Mathematics and Theoretical Physics,  University of Cambridge, \protect\\ Cambridge CB3 0WA, United Kingdom}
\affil[3]{\small Univ Lyon, Inria, ENS Lyon, UCBL, LIP, Lyon, France}

\date{}

\maketitle

\begin{abstract}
{Uhlmann's theorem is a cornerstone of quantum information theory, stating that for any quantum state $\rho_{AB}$ and any state $\sigma_A$, there exists an extension $\sigma_{AB}$ of $\sigma_A$ such that the fidelity between $\rho_{AB}$ and $\sigma_{AB}$ equals the fidelity between their marginals $\rho_A$ and $\sigma_A$. This property underpins many results and applications in quantum information science. In this work, we generalize Uhlmann's theorem to a broad class of measured $f$-divergences, including the measured $\alpha$-R\'enyi divergences for all $\alpha \geq 0$. The well-known Uhlmann's theorem for the fidelity corresponds to the special case $\alpha = \frac{1}{2}$. Since most commonly used quantum \Renyi divergences, including the Petz and sandwiched \Renyi divergences, cannot satisfy this property (except for degenerate cases). This fundamentally distinguishes measured $f$-divergences from other quantum divergences and highlights their unique mathematical structure.}
\end{abstract}

\newcommand{\Sm}{S}

\section{Introduction}

For a convex function $f : \RR_{>0} \to \RR$, the $f$-divergence  between discrete probability vectors $P$ and $Q$ is defined as\footnote{We make the following standard conventions to define this expression: $f(0) = \lim_{t \downarrow 0} f(t)$, $0 f(\frac{0}{0}) = 0$ and $0 f(\frac{a}{0}) = \lim_{t \downarrow 0} t f(\frac{a}{t})$ for $a > 0$. Note that the limits may be $+\infty$.} 
\begin{align}
\label{eq: f divergence}
    \Sm_{f}(P \| Q) := \sum_{x} Q(x) f\left(\frac{P(x)}{Q(x)}\right).
\end{align}

This family already appears in R\'enyi's work~\cite{renyi1961measures} and was further studied by Csisz\'ar~\cite{csiszar1967information} and Ali and Silvey~\cite{ali1966general}.
Important examples include the Kullback-Leibler divergence (also called relative entropy) $D( \cdot \| \cdot )$ corresponding to $f_1(t) = t \log t$ with $D(P \| Q) = \Sm_{f_1}(P \| Q)$ and the $\alpha$-R\'enyi divergences $D_{\alpha}( \cdot \| \cdot )$ corresponding to  $f_{\alpha}(t) = \sign(\alpha-1)t^{\alpha}$ with $D_{\alpha}(P \| Q) = \frac{1}{\alpha - 1} \log \left( \sign(\alpha - 1) \Sm_{f_{\alpha}}(P \| Q) \right)$. Such divergences play an important role in various areas of statistics and information theory, we refer to~\cite[Chapter 7]{polyanskiy2025information} for an overview. 

For quantum states modeled by positive semidefinite operators $\rho$ and $\sigma$, there are multiple ways of extending~\eqref{eq: f divergence}; see~\cite{hiai2017different} for a systematic exposition. The main focus in this paper is on the \emph{measured}~\cite{donald1986relative,hiai1991proper} $f$-divergence which is obtained by performing a measurement $M$ on the states $\rho$ and $\sigma$ and computing the $f$-divergence between the resulting probability distributions $P_{\rho, M}$ and $P_{\sigma, M}$. We then take the supremum over all possible measurements or all projective measurements:
\begin{align}
\label{eq: DM def POVM}
    \Sm_{\Meas,f}(\rho \| \sigma) & := \sup_{M \text{ POVM}} \Sm_f(P_{\rho, M} \| P_{\sigma, M}), \\
    \Sm_{\Meas,f}^{\Proj}(\rho \| \sigma) & := \sup_{M \text{ PVM}} \Sm_{f}(P_{\rho,M} \| P_{\sigma, M}),
    \label{eq: DM def PVM}
\end{align}
where in~\eqref{eq: DM def POVM} the supremum is over positive operator-valued measures (POVM) $M$ described by a discrete set $\cX$ and positive semidefinite operators $\{M_x\}_{x \in \cX}$ satisfying $\sum_{x \in \cX} M_x = I$, the supremum in~\eqref{eq: DM def PVM} is over projection-valued measures (PVM) $M$, i.e., POVMs satisfying in addition that $M_x$ is an orthogonal projector for all $x \in \cX$ and $P_{\rho, M}(x) = \tr[\rho M_x]$.

\textbf{Main results:}  We consider the Uhlmann property that is well known for the fidelity~\cite{uhlmann1976transition}: given a quantum state $\rho_{A}$ on $A$ and a state $\sigma_{AR}$ on $A \otimes R$, is there an extension $\rho_{AR}$ of $\rho_{A}$ such that $\Sm_{\Meas,f}(\rho_{AR} \| \sigma_{AR}) = \Sm_{\Meas,f}(\rho_{A} \| \sigma_{A})$? We show in Theorem~\ref{thm: DM f Uhlmann} that if $f^*$ is operator convex and operator monotone and the domain of $f^*$ is unbounded from below, then Uhlmann's property holds. Similarly, if $(f^*)^{-1}$ is operator concave and operator monotone and the range of $f^*$ is unbounded from above, then Uhlmann's property holds with the roles of $\rho$ and $\sigma$ exchanged. In the special case of R\'enyi divergences, this shows the Uhlmann property for all $\alpha \geq 0$ (Corollary~\ref{thm: DM Uhlmann}). This was previously known for $\alpha = \frac{1}{2}$~\cite{uhlmann1976transition} and in the limit $\alpha \to \infty$~\cite{metger2024generalised}. Our result can also be used to recover the regularized Uhlmann's theorem of~\cite{metger2024generalised}. We conclude this paper in Section~\ref{sec: duality} with an intriguing ``duality'' property (Proposition~\ref{thm: duality sum zero}) {relating the Umegaki relative entropy $D(\rho \| \cvxset)$ between $\rho$ and a compact convex set $\cvxset$ and the measured relative entropy $D_{\Meas}(\rho \| \cvxset_{\pl \pl}^{\circ})$ between $\rho$ and the corresponding polar set $\cvxset_{\pl \pl}^{\circ}$.
A similar duality holds between $D(\rho\|\sC^\circ_{\pl \pl})$ and $D_{\Meas}(\rho\|\sC)$. These duality relations provide a clearer understanding of why the measured relative entropy to a set of quantum states becomes superadditive when the polar sets are closed under tensor products. This superadditivity property is fundamental for several recent applications, including~\cite{fang2024generalized,fang2025efficient,fang2025adversarial,arqand2025marginal}.} {The significance and some potential applications of our results are discussed in Section~\ref{sec: discussions}.}

\textbf{Notation:} For a finite-dimensional Hilbert space $A$, we let $\HERM(A)$ denote the space of Hermitian operators on $A$, $\PSD(A) = \{\omega \in \HERM(A) : \omega \geq 0\}$ the set of positive semidefinite operators and $\PD(A) = \{\omega \in \HERM(A) : \omega > 0\}$ the set of positive definite operators. The set of density operators is denoted by $\density(A) = \{ \rho \in \PSD(A) : \tr[\rho] = 1\}$. We drop the Hilbert space from the notation when it is clear from the context. For $\omega \in \HERM$, $\spec(\omega)$ denotes the spectrum of $\omega$. The measured R\'enyi divergences for $\alpha \geq 0$ are defined as
\begin{align}
    D_{\Meas,\alpha}(\rho \| \sigma) 
    &= \frac{1}{\alpha - 1} \log \left( \sign(\alpha - 1) \Sm_{\Meas,f_\alpha}(\rho \| \sigma) \right), & \text{ for $\alpha \neq 1$} \\
    D_{\Meas}(\rho \| \sigma) &= S_{\Meas, f_1}( \rho \| \sigma), & \text{ for $\alpha = 1$.} 
\end{align}
In the limit $\alpha \to \infty$, it is known that $D_{\Meas,\alpha} \to D_{\max}$, see e.g.,~\cite[Appendix A]{mosonyi2015quantum}.

\textbf{Independent work:} We note that in independent and concurrent work, Mazzola, Renner and Sutter~\cite{mazzola2025uhlmann} established a regularized Uhlmann's theorem for the sandwiched R\'enyi divergence for $\alpha \geq \frac{1}{2}$ as well as an Uhlmann inequality for the measured R\'enyi divergences i.e., they show the existence of an extension satisfying  $D_{\Meas,\alpha}(\rho_{AR} \| \sigma_{AR})$ is upper bounded by the corresponding sandwiched R\'enyi divergence.

\section{Variational expression for measured divergences}

In the classical setting, variational expressions for $f$-divergences play an important role (see \cite[Chapter 7]{polyanskiy2025information}), but they are also fundamental in the quantum case in particular for the measured divergences as they are at the core of many applications of such divergences; see e.g.,~\cite{Berta2017,sutter2017multivariate,fang2024generalized}. {In this section, we discuss the variational expression for measured $f$-divergences, which forms an important starting point for finding the Uhlmann's theorem for measured divergences in the next section. These variational expressions have been studied in the von Neumann algebra setting in~\cite[Chapter 5]{hiai2021quantum}. Here, we present the results and their proofs in the finite-dimensional setting for completeness and for accessiblility to general readers.}

In order to state the variational expression, we need to introduce the Fenchel conjugate of a function. For a function $f : \RR \to \RR \cup \{+\infty\}$, the Fenchel conjugate (also called Legendre transform) is defined as~\cite{boyd2004convex} 
\begin{equation}
    \label{eq:fconj}
    f^*(z) := \sup_{t \in \dom(f)} \, \{ tz - f(t) \},
\end{equation}
where $\dom(f) = \{t \in \RR : f(t) < +\infty\}$. The domain of $f^*$ is defined in the same way as
\[
\dom(f^*) = \{z \in \RR : f^*(z) < +\infty \}.
\]

As $f$ should define a divergence, we assume throughout this paper that $(0,\infty) \subset \dom(f) $ and that $f$ is convex. We assume in addition that $f$ is lower semicontinuous. Note that the $f$-divergence only depends on the values $f$ takes on $(0,\infty)$, but the Fenchel conjugate (and hence the variational expressions) depends on the values of $f$ on $\dom(f)$.

We present the following general variational formula for the projective measured $f$-divergence under the sole condition that $f$ is convex and lower semicontinuous. {This result corresponds to~\cite[Theorem 5.7]{hiai2021quantum}.} It can also be seen as a quantum analog of~\cite[Theorem 7.26]{polyanskiy2025information}. 
\begin{shaded}
\begin{proposition}\label{prop: general variational formula}
Let $f: \RR \to \RR\cup \{+\infty\}$ be a convex and lower semicontinuous function and $(0,\infty) \subset \dom(f)$. For any $\rho , \sigma \in \PSD$, 
\begin{align}\label{eq: general variational formula}
  \Sm_{\Meas,f}^{\Proj}(\rho \| \sigma) = \sup_{\omega \in \HERM,\; \spec(\omega) \subset \dom(f^*)}  \tr[\rho \omega] - \tr[\sigma f^*(\omega)],
\end{align}
where $f^*$ is the Fenchel conjugate of $f$ as defined in \eqref{eq:fconj}.
\end{proposition}
\end{shaded}
\begin{proof}
Since $f$ is convex and lower semicontinuous, we know that $f = (f^{*})^*$~\cite[Theorem 12.2]{rockafellar}, i.e., for any {$t \in \RR$}
\begin{equation}
\label{eq:var-fenchel-proof}
f(t) = \sup_{z \in \dom(f^*)} \{tz - f^*(z)\}.
\end{equation}
Hence for any $r,s > 0$ we get
\[
sf(r/s) = \sup_{z \in \dom(f^*)} \{ rz - sf^*(z) \}.
\]
{This expression also matches the conventions we took when $r = 0$ or $s = 0$. In fact, if $r = 0, s = 1$, 
we use \eqref{eq:var-fenchel-proof} to get $f(0) = \sup_{z \in \dom(f^*)}\{-f^*(z)\}$. If $r = s = 0$, both sides are equal to $0$.  For $r > 0$ and $s = 0$, we need to show that $\lim_{s \downarrow 0} s f(r/s) = \sup_{z \in \dom(f^*)} \{ rz\}$. By performing the change of variable $x = r/s$, this is equivalent to showing that $\lim_{x \to \infty} \frac{f(x)}{x} = \sup \dom(f^*)$. To establish this equality, we prove the two inequalities separately. We start with ($\geq$): we have $\frac{f(x)}{x} = \sup_{z \in \dom(f^*)} \{ z - \frac{f^*(z)}{x} \} \geq z - \frac{f^*(z)}{x}$ for any $z \in \dom(f^*)$. Thus, $\lim_{x \to \infty} \frac{f(x)}{x} \geq \sup \dom(f^*)$. Now for ($\leq$), let $w < \lim_{x \to \infty} \frac{f(x)}{x}$ and let $x_{n} \geq n$ be such that $f(x) \geq w x$ for all $x \geq x_{n}$ ($n$ is an integer $n \geq 2$). Let $y_n$ be any subderivative of $f$ at $x_{n}$. By definition, for any $x \in \RR$, $f(x) \geq f(x_{n}) + y_n(x-x_n)$. This implies that $y_n x - f(x) \leq y_n x_n - f(x_n)$ and thus $f^*(y_n) \leq y_n x_n - f(x_n) < +\infty$ so $y_n \in \dom(f^*)$. In addition, we have $y_n \geq \frac{f(x_n) - f(1)}{x_n - 1} \geq w \frac{x_{n}}{x_{n} - 1} - \frac{f(1)}{x_n - 1}$ and $1 \in \dom(f)$ by assumption. As a result, for any $\epsilon > 0$, there exists an $n$ such that we obtain $y_n \geq w - \epsilon$. Thus $\sup \dom(f^*) \geq w$ which concludes the second inequality. 
\footnote{We thank an anonymous reviewer for suggesting this argument.
}}
Plugging into~\eqref{eq: DM def PVM}, we get \begin{align}
  \Sm_{\Meas,f}^{\Proj}(\rho \| \sigma)
  &= \sup_{\{M_x\} \text{ PVM}} \sum_{x} \tr[M_x \sigma] f\left(\frac{\tr[M_x \rho]}{\tr[M_x \sigma]}\right) \\
  & = \sup_{\{M_x\} \text{ PVM}} \: \sum_x \sup_{z_x \in \dom(f^*)} \tr[M_x \rho ] z_x - \tr[M_x \sigma] f^*(z_x)
   \label{eq: M P equiv tmp1}\\
  & = \sup_{\{M_x\} \text{ PVM}} \:  \sup_{\{z_x\}} \sum_x \tr[z_x M_x\rho] - \tr[f^*(z_x) M_x \sigma] \\
  & = \sup_{\{M_x\} \text{ PVM}} \:  \sup_{\{z_x\} } \tr \left[ \sum_x z_xM_x \rho \right] - \tr \left[ f^*\left(\sum_x z_x M_x\right) \sigma \right] \label{eq: M P equiv tmp2}\\
  & = \sup_{\omega \in \HERM, \spec(\omega) \subset \dom(f^*)} \tr[\rho \omega] - \tr[\sigma f^*(\omega)],
\end{align}
where everywhere $\{M_x\}$ is a projective measurement and $z_x \in \dom(f^*)$. Note that in the third line we used the fact that $M_x$ are orthogonal projectors.
\end{proof}

Table \ref{tbl:fenchel} shows the Fenchel conjugate of some common functions.

\begin{table}
\footnotesize
\centering
{\setlength\extrarowheight{2pt}
\begin{tabular}{llllp{2.5cm}p{2.5cm}}
\toprule
$f(t)$ & $\dom(f)$ & $f^*(z)$ & $\dom(f^*)$ & Props. of $f^*$\newline on $\dom(f^*)$ & Props. of $(f^*)^{-1}$\newline on $\range(f^*)$\\ \midrule
$-t^{\alpha}$ ($\alpha \in [0,1)$) & $\RR_{>0}$ & $(1-\alpha) \left(\frac{-z}{\alpha}\right)^{-\frac{\alpha}{1-\alpha}}$ & $\RR_{<0}$ & 
$\alpha \in [0,\frac{1}{2}]$: op. convex and monotone
& $\alpha \in [\frac{1}{2},1)$: op. concave and monotone
\\ \hline
$\begin{cases}
    t^{\alpha} & t \geq 0\\
    0 & t < 0
\end{cases}$ \;\; ($\alpha > 1$) & $\RR$ & $(\alpha-1)\left(\frac{z}{\alpha}\right)^{\frac{\alpha}{\alpha-1}}$ & $\RR_{>0}$ &
$\alpha \geq 2$: op. convex
& 
$\alpha > 1$: op. concave and monotone
\\ \hline
$t\log(t)$ & $\RR_{>0}$ & $e^{z-1}$ & $\RR$ & Not applicable & Op. concave and monotone \\ \hline
$\frac{1}{2}|t-1|$ & $\RR$ & $z$ & $[-\frac{1}{2},\frac{1}{2}]$ & Linear monotone & Linear monotone\\
\bottomrule
\end{tabular}}
\caption{Convex functions $f$ and their Fenchel conjugates}
\label{tbl:fenchel}
\end{table}

\begin{remark}
\label{rem: examples DM proj}
We can recover \Renyi divergences for $\alpha \in [0,1)$ by choosing
\[
f_{\alpha}(t) = -t^{\alpha}, \qquad \dom(f_{\alpha}) = \RR_{>0}
\]
whose Fenchel conjugate is given in Table \ref{tbl:fenchel}. In this case, the variational expression \eqref{eq: general variational formula} becomes:
\[
  \Sm_{\Meas,f_{\alpha}}^{\Proj}(\rho \| \sigma) = \sup_{\omega < 0}  \tr[\rho \omega] - (1-\alpha) \tr\left[\sigma \left(-\frac{\omega}{\alpha}\right)^{-\frac{\alpha}{1-\alpha}}\right],
\]
which becomes after the change of variable $\gamma = -\frac{\omega}{\alpha}$:
\[
  \Sm_{\Meas,f_{\alpha}}^{\Proj}(\rho \| \sigma) = \sup_{\gamma > 0}  -\alpha \tr[\rho \gamma] - (1-\alpha) \tr\left[\sigma \gamma^{-\frac{\alpha}{1-\alpha}}\right].
\]
Then, since $D_{\Meas, \alpha}(\rho \| \sigma) = \frac{1}{\alpha-1} \log\left( -\Sm^{\Proj}_{\Meas,f_{\alpha}}(\rho\|\sigma)\right)$, we get the variational expression established in~\cite{Berta2017} for $\alpha \in (0,1)$ (up to a simple change of variable that is described in Remark~\ref{rem: examples DM equality}). Note that for $\alpha = 0$, we interpret $\gamma^0$ as the projector on the support of $\gamma$.

When $\alpha > 1$, we consider the function\footnote{The reason we consider the function \eqref{eq:falpha>1} defined on the whole of $\RR$ instead of just $\RR_{>0}$ is that the resulting Fenchel conjugate is defined on $\RR_{>0}$ instead of the whole of $\RR$ which is useful in the application of Theorem \ref{thm: POVM equals PVM}}
\begin{equation}
    \label{eq:falpha>1}
    f_{\alpha}(t) = \begin{cases} t^{\alpha} & t \geq 0\\
0 & t < 0,
\end{cases}
\qquad \dom(f_{\alpha}) = \RR
\end{equation}
whose Fenchel conjugate is given in Table \ref{tbl:fenchel}. The expression~\eqref{eq: general variational formula} gives 
\[
\Sm_{\Meas,f_{\alpha}}^{\Proj}(\rho \| \sigma) = \sup_{\omega > 0}  \alpha \tr[\rho \omega] + (1-\alpha) \tr\left[\sigma \omega^{\frac{\alpha}{1-\alpha}}\right],
\]
which also matches with the expression found in~\cite{Berta2017}.

For the measured {relative entropy}, we choose $f_1(t) = t \log t$ and $\dom(f_1) = \RR_{>0}$. This gives
\begin{align}
\label{eq:kl measured}
\Sm_{\Meas,f_{1}}^{\Proj}(\rho \|\sigma) = \sup_{\omega \in \HERM}  \tr[\rho \omega] - \tr\left[\sigma e^{\omega - 1}\right],
\end{align}
which again is equivalent to the expressions obtained in~\cite{petzquantum,Berta2017}.
\end{remark}

The variational expression \eqref{eq: general variational formula} holds for all convex lower semicontinuous functions $f$. However, the optimization program as written is in general not a convex one unless additional conditions are imposed, such as operator convexity of $f^*$. The theorem below gives a general sufficient condition under which $S_{\Meas,f}^{\Proj}(\rho\|\sigma)$ can be expressed as a convex program and which guarantees that {the measured divergence and the one restricted to projective measurements match.} {This is similar to the result of~\cite[Theorem 5.8]{hiai2021quantum}.}

\begin{shaded}
\begin{theorem}\label{thm: POVM equals PVM}
    Let $f:\RR\to \RR\cup\{+\infty\}$ be a convex and lower semicontinuous function such that $(0,\infty) \subset \dom(f)$, and let $f^*$ be its Fenchel conjugate with domain $\dom(f^*)$.
    
    Assume $\psi:J\to \dom(f^*)$ is a one-to-one map from an interval $J \subset \RR$ to $\dom(f^*)$ such that
    \begin{itemize}
        \item $\psi$ is operator concave on $J$,
        \item $f^* \circ \psi$ is operator convex on $J$.
    \end{itemize}
    Then we have, for any $\rho,\sigma \in \PSD$
    \begin{equation}
    \label{eq:SM variational formula change of variables}
\Sm_{\Meas,f}(\rho \| \sigma) = \Sm_{\Meas,f}^{\Proj}(\rho \| \sigma) = \sup_{\gamma \in \HERM,\; \spec(\gamma) \subset J} \left\{ \tr[\rho \psi(\gamma)] - \tr[\sigma (f^*\circ\psi)(\gamma)] \right\}.
    \end{equation}
\end{theorem}
\end{shaded}

\begin{remark}\label{rmk: DPI}
Combining Theorem~\ref{thm: POVM equals PVM} with ~\cite[Corollary 4.20]{hiai2017different}, it follows that under the conditions of the theorem, {the data-processing inequality of $S_{\Meas,f}$ holds under any positive trace-preserving maps.}
\end{remark}

\begin{remark}
\label{rem: examples DM equality}
    For all the examples mentioned in Remark~\ref{rem: examples DM proj}, such a $\psi$ can be found. In fact, for $\alpha \in (0,\frac{1}{2}]$, $f_{\alpha}^*(z) = (1-\alpha) (\frac{-z}{\alpha})^{-\frac{\alpha}{1-\alpha}}$ is already operator convex so $\psi(z) = z$ works. For $\alpha \in [\frac{1}{2}, 1)$, we can choose $J = \RR_{>0}$ and $\psi(\lambda) = (f_{\alpha}^*)^{-1}(\lambda) = -\alpha (\frac{\lambda}{1-\alpha})^{-\frac{1-\alpha}{\alpha}}$. For $\alpha > 1$, we can also choose $J = \RR_{>0}$ and $\psi(\lambda) = (f_{\alpha}^*)^{-1}(\lambda) =\alpha (\frac{\lambda}{\alpha-1})^{\frac{\alpha-1}{\alpha}}$. Similarly, for $f(t)=t\log t$, we choose $J = \RR_{>0}$ and $\psi(\lambda) = (f^*)^{-1}(\lambda) = 1 + \log \lambda$. For all these examples, the result~\eqref{eq:SM variational formula change of variables} was already established in~\cite{Berta2017}. 
    
    Another example is the measured total variation distance which is obtained by taking $f_{TV}(t) = \frac{1}{2} |t-1|$ and $\dom(f) = \RR$. With this choice $f_{TV}^*(z) = z$ for $z \in \dom(f_{TV}^*) = [-\frac{1}{2},\frac{1}{2}]$. Expression~\eqref{eq:SM variational formula change of variables} then gives the well-known variational formulation of the trace distance.
\end{remark}

\begin{proof}[Proof of Theorem \ref{thm: POVM equals PVM}]
The second equality in~\eqref{eq:SM variational formula change of variables} is a simple change of variables to \eqref{eq: general variational formula}, namely $\omega = \psi(\gamma)$.
The inequality $\Sm_{\Meas,f}(\rho \| \sigma) \geq \Sm_{\Meas,f}^{\Proj}(\rho \| \sigma)$ is clear and so it remains to show the reverse inequality.
Fix $\{M_x\}$ any POVM and note that we can write, using the same argument as in \eqref{eq: M P equiv tmp1}
\[
\sum_{x} \tr[M_x \sigma] f\left(\frac{\tr[M_x \rho]}{\tr[M_x \sigma]}\right) = \sup_{\{z_x\} \subset \dom(f^*)} \tr\left[\sum_x z_x M_x \rho\right] - \tr\left[\sum_x f^*(z_x) M_x \sigma\right].
\]
Since $\psi$ is a one-to-one map from $J$ to $\dom(f^*)$, we can do the change of variables $z_x = \psi(\lambda_x)$ where $\lambda_x \in J$, and so we get
\begin{equation}
\label{eq: POVM psi 11}
\sum_{x} \tr[M_x \sigma] f\left(\frac{\tr[M_x \rho]}{\tr[M_x \sigma]}\right) = \sup_{\{\lambda_x\} \subset J} \tr\left[\sum_x \psi(\lambda_x) M_x \rho\right] - \tr\left[\sum_x f^*(\psi(\lambda_x)) M_x \sigma\right].
\end{equation}
By assumption, $\psi$ is operator concave, and so using the operator Jensen inequality~\cite{hansen2003jensen}
\[
\psi\left(\sum_x \lambda_x M_x\right) = \psi\left(\sum_x \sqrt{M_x} (\lambda_x I) \sqrt{M_x}\right) \geq \sum_{x} \psi(\lambda_x) M_x,
\]
and similarly, since $f^* \circ \psi$ is operator convex we have
\[
\sum_x f^*(\psi(\lambda_x)) M_x \geq (f^* \circ \psi) \left( \sum_x \lambda_x M_x \right).
\]
Combining this relation with \eqref{eq: POVM psi 11} gives
\[
\begin{aligned}
\sum_{x} \tr[M_x \sigma] f\left(\frac{\tr[M_x \rho]}{\tr[M_x \sigma]}\right) &\leq \sup_{\{\lambda_x\} \subset J} \tr \left[ \psi\left(\sum_x \lambda_x M_x\right) \rho \right] - \tr \left[ f^* \circ \psi\left( \sum_x \lambda_x M_x \right) \sigma \right]\\
& \leq \sup_{\gamma \in \HERM,\; \spec(\gamma) \subset J} \tr [ \psi(\gamma) \rho ] - \tr[ (f^* \circ \psi)(\gamma) \sigma ]\\
&= \Sm^{\Proj}_{\Meas,f}(\rho \| \sigma)
\end{aligned}
\]
where in the second inequality above we used $\gamma = \sum_{x} \lambda_x M_x$ whose spectrum is in $J$ since $M_x \geq 0$ and $\sum_x M_x = I$.
\end{proof}

\section{Uhlmann's theorem for measured divergences}

Uhlmann’s theorem is a fundamental result in quantum information theory that describes the relationship between the fidelity of mixed quantum states and their extensions~\cite{uhlmann1976transition}. More specifically, it states that for any two quantum states $\rho_{A}$ and $\sigma_{AR}$, there exists an extension $\rho_{AR}$ of $\rho_{A}$ such that the fidelity $F(\rho_{AR},\sigma_{RA})$ equals the fidelity $F(\rho_A,\sigma_A)$. The following result generalizes the renowned Uhlmann’s theorem to measured $f$-divergences in general.

\begin{shaded}
\begin{theorem}[Uhlmann's theorem for measured $f$-divergences]\label{thm: DM f Uhlmann}
Let $f$ be a convex lower semicontinuous function {such that $(0,\infty) \subset \dom(f)$}.
 If its Fenchel conjugate $f^*$ is operator convex and operator monotone on $\dom(f^*)$ and $\inf \{z : z \in \dom(f^*)\} = -\infty$, then for any $\rho_A \in \density(A)$ and $\sigma_{AR} \in \PSD(AR)$, we have
  \begin{align}
 \label{eq: thm uhlmann statement 1}
     \inf_{\substack{\rho_{AR} \in \PSD(AR) \\ \tr_{R} \rho_{AR} = \rho_A}} \Sm_{\Meas, f}(\rho_{AR} \| \sigma_{AR}) = \Sm_{\Meas,f}(\rho_{A} \| \sigma_A).
 \end{align}
 Similarly, if $f^* : \dom(f^*) \to \range(f^*)$ is one-to-one and $(f^*)^{-1}$ is operator concave and operator monotone on $\range(f^*)$ and $\sup \{\lambda : \lambda \in \range(f^*)\} = +\infty$, then for any $\rho_{AR} \in \density(AR)$ and $\sigma_A \in \PSD(A)$, we have
 \begin{align}
 \label{eq: thm uhlmann statement}
     \inf_{\substack{\sigma_{AR} \in \PSD(AR) \\ \tr_{R} \sigma_{AR} = \sigma_A}} \Sm_{\Meas, f}(\rho_{AR} \| \sigma_{AR}) = \Sm_{\Meas,f}(\rho_{A} \| \sigma_A).
 \end{align}
 In addition, the infimum in both equations is attained.
 \end{theorem}
 \end{shaded}
 \begin{remark}
    We leave it as an open question whether such a result can be extended to other choices of $f$. However, {only assuming that $f^*$ is operator convex and monotone on $\dom(f^*)$ is not sufficient.}
    For example, the trace distance does not satisfy the Uhlmann property, but as mentioned in Remark~\ref{rem: examples DM equality}, we have $f^*_{TV}(z) = z$ for $z \in \dom(f^*) = [-\frac{1}{2},+\frac{1}{2}]$ which is clearly operator convex and monotone. 
 \end{remark}
\begin{remark}{
    Let $\phi_{RA}$ be a purification of $\rho_A$. Then the set of all extensions of $\rho_A$ can be obtained by ranging over all completely positive trace-preserving (CPTP) maps $\mathcal{E}_{R \to R}$ acting on the purifying system $R$, i.e., 
    \begin{align*} 
        \{\rho_{AR} \in \PSD(RA): \tr_R \rho_{AR} = \rho_A\} =\{(\id_A \otimes \mathcal{E}_{R \to R})(\phi_{AR}) : \mathcal{E}_{R \to R} \in \text{ CPTP}\}.
    \end{align*} 
    This can be similarly proved as~\cite[Lemma 10]{fang2019quantum}. More explicitly, it is clear that ``$\supset$'' holds. Now for any extension $\rho_{AR}$ of $\rho_A$. We have $\bar \rho_{AR} = \rho_A^{-1/2} \rho_{AR} \rho_A^{-1/2}$. Then we have $\bar \rho_{AR} \geq 0$ and $\tr_{R} \bar \rho_{AR} = I_A$. From the Choi-Jamiolkowski isomorphism, we know that there exists a CPTP map $\cN_{R\to R}$ such that $\bar \rho_{AR} = \id_A\ox \cN_{R\to R}(\Phi_{RA})$, where $\Phi_{RA}$ denotes the unnormalized maximally entangled state. Thus, we get $\rho_{AR} = \id_A\ox \cN_{R\to R}(\rho_A^{1/2}\Phi_{RA}\rho_A^{1/2})$. Denote $\psi_{RA} := \rho_A^{1/2}\Phi_{RA}\rho_A^{1/2}$. We know that $\psi_{RA}$ is a purification of $\rho_A$. Due to the isometric equivalence between purifications, there exists an isometry $\cU_{R\to R}$ on the system $R$ such that $\psi_{RA} = \id_A \ox \cU_{R\to R}(\phi_{RA})$. This gives $\rho_{AR} = \id_A \ox \cN_{R\to R} \circ \cU_{R\to R}(\phi_{RA})$. So we have the ``$\subset$'' direction.  
    As a result, Eq.~\eqref{eq: thm uhlmann statement 1} can be equivalently reformulated in terms of purifications and channels. Let $\phi_{AR}$ be any purification of $\rho_A$. Then we have
    \begin{align*}
    \inf_{\mathcal{E} \in \text{CPTP}} \Sm_{\Meas, f}((\id_A \otimes \mathcal{E}_{R \to R})(\phi_{AR}) \| \sigma_{AR}) = \Sm_{\Meas,f}(\rho_{A} \| \sigma_A),
    \end{align*}
    where the infimum is over all CPTP maps $\mathcal{E}_{R \to R}$ acting on the purifying system $R$. When $\sigma_{AR}$ is a pure state, this formulation closely parallels the original Uhlmann's theorem for the fidelity~\cite{uhlmann1976transition}, which involves optimization over isometries on the purifying system. 
    Similarly, Eq.~\eqref{eq: thm uhlmann statement} can be reformulated as: for any purification $\psi_{AR}$ of $\sigma_A$, we have
    \begin{align*}
    \inf_{\mathcal{E} \in \text{CPTP}} \Sm_{\Meas, f}(\rho_{AR} \| (\id_A \otimes \mathcal{E}_{R \to R})(\psi_{AR})) = \Sm_{\Meas,f}(\rho_{A} \| \sigma_A).
    \end{align*}}
\end{remark}

\begin{proof}
 We establish the result for Eq.~\eqref{eq: thm uhlmann statement 1} first. 
 Recall the variational expression in Eq.~\eqref{eq:SM variational formula change of variables} with $\psi = \id$:
 \begin{align}
    \Sm_{\Meas,f}(\rho \| \sigma) = \sup_{\substack{\gamma \in \HERM\\ \spec(\gamma) \subset \dom(f^*)}} \left\{ \tr[\rho \gamma] - \tr[\sigma f^*(\gamma)] \right\}.
     \label{eq: var expression Qalpha proof uhlmann}
 \end{align}
We now use Sion's minimax theorem~\cite[Corollary 3.3]{Sion1958} to get 
\begin{align*}
    \inf_{\substack{\rho_{AR} \in \PSD(A R) \\ \tr_{R} \rho_{AR} = \rho_A}} \Sm_{\Meas, f}(\rho_{AR} \| \sigma_{AR}) =  \sup_{\substack{\gamma_{AR} \in \HERM(AR) \\ \spec(\gamma_{AR}) \subset \dom(f^*)}} \inf_{\substack{\rho_{AR} \in \PSD(A R) \\ \tr_{R} \rho_{AR} = \rho_A}} \left\{ \tr[\rho_{AR} \gamma_{AR}] - \tr[\sigma_{AR} f^*(\gamma_{AR})] \right\}.
\end{align*}
In fact, the set $\{\rho_{AR} \in \PSD(A R) : \tr_{R} \rho_{AR} = \rho_A\}$ is compact and convex and $\{\gamma_{AR} \in \HERM(AR): \spec(\gamma_{AR}) \subset \dom(f^*)\}$ is convex. In addition the expression $\tr[\rho_{AR} \gamma_{AR}] - \tr[\sigma_{AR} f^*(\gamma_{AR})]$ is concave in {$\gamma_{AR}$} (because $f^*$ is operator convex) and linear in {$\rho_{AR}$}. Now, by semidefinite programming duality {(e.g.~\cite[Section 1.2.3]{watrous2018theory})}, we have for a fixed $\gamma_{AR} \in \HERM(AR)$
\begin{align}
    \inf_{\substack{\rho_{AR} \in \PSD(A R) \\ \tr_{R} \rho_{AR} = \rho_A}} \tr[\rho_{AR} \gamma_{AR}] = \sup_{\substack{\Lambda_A \in \HERM(A) \\ \Lambda_A \otimes I_{R} \leq \gamma_{AR}}} \tr[\rho_A \Lambda_A].
\end{align}
In fact it is clear that the infimum is finite as $0 \leq \rho_{AR} \leq \tr[\rho_{A}] I_{AR}$ and $\Lambda_A = -(\| \gamma_{AR} \|_{\infty} + 1) I_{A}$ is strictly feasible for the dual. As a result,
\begin{align}
    \inf_{\substack{\rho_{AR} \in \PSD(A R) \\ \tr_{R} \rho_{AR} = \rho_A}} \Sm_{\Meas, f}(\rho_{AR} \| \sigma_{AR}) =  \sup_{\substack{\gamma_{AR} \in \HERM(AR)\\ \spec(\gamma_{AR}) \subset \dom(f^*)}} \sup_{\substack{\Lambda_A \in \HERM(A) \\ \Lambda_A \otimes I_{R} \leq \gamma_{AR}}} \left\{\tr[\rho_A \Lambda_A] - \tr[\sigma_{AR} f^*(\gamma_{AR})] \right\}.
\end{align}
Note that we can obtain a lower bound on the supremum (the easy direction) by restricting ourselves to $\gamma_{AR} = \Lambda_A \otimes I_{R}$ and observing that $f^*(\Lambda_A \otimes I_R) = f^*(\Lambda_A) \otimes I_R$:
\begin{align}
     \inf_{\substack{\rho_{AR} \in \PSD(A R) \\ \tr_{R} \rho_{AR} = \rho_A}} \Sm_{\Meas, f}(\rho_{AR} \| \sigma_{AR})
    &\geq \sup_{\substack{\Lambda_A \in \HERM(A) \\ \spec(\Lambda_A) \subset \dom(f^*)}} \left\{\tr[\rho_A \Lambda_A] - \tr[\sigma_{A} f^*(\Lambda_{A})] \right\}. \\
    &= \Sm_{\Meas,f}(\rho_{A} \| \sigma_{A}).
\end{align}
For the other direction, as $f^*$ is operator monotone we have for any feasible {$(\gamma_{AR}, \Lambda_A)$}, it holds that $f^*(\gamma_{AR}) \geq f^*(\Lambda_A)\ox I_R$. This then implies
\begin{align}
    \inf_{\substack{\rho_{AR} \in \PSD(A R) \\ \tr_{R} \rho_{AR} = \rho_A}} \Sm_{\Meas, f}(\rho_{AR} \| \sigma_{AR})
    &\leq \sup_{\substack{\gamma_{AR} \in \HERM(AR) \\ \spec(\gamma_{AR}) \subset \dom(f^*)}} \sup_{\substack{\Lambda_A \in \HERM(A) \\ \Lambda_A \otimes I_{R} \leq \gamma_{AR}}} \left\{\tr[\rho_A \Lambda_A] - \tr[\sigma_{A} f^*(\Lambda_A)] \right\} \\
    &\leq \sup_{\substack{\Lambda_A \in \HERM(A) \\ \spec(\Lambda_{A}) \subset \dom(f^*)}} \left\{\tr[\rho_A \Lambda_A] - \tr[\sigma_{A} f^*(\Lambda_{A})] \right\}. \\
    &= {\Sm_{\Meas,f}(\rho_A \| \sigma_A)}.
\end{align}
In the second inequality, we used the condition $\inf \dom(f^*) = -\infty$ to get that $\Lambda_{A} \otimes I_{R} \leq \gamma_{AR}$ and $\spec(\gamma_{AR}) \subset \dom(f^*)$ implies $\spec(\Lambda_{A}) \subset \dom(f^*)$.
This establishes~\eqref{eq: thm uhlmann statement 1}. 

To prove that the infimum is attained, {it is well-known that a lower semicontinuous function attains its infimum on a compact set~\cite[Theorem 7.3.1]{kurdila2005convex}. This can be more explicitly argued as follows.} For any $n \geq 1$, we can take $\rho_{AR}^{(n)} \in \PSD(AR)$ satisfying $\tr_{R} \rho_{AR}^{(n)} = \rho_{A}$ such that $0 \leq \Sm_{\Meas, f}(\rho^{(n)}_{AR} \| \sigma_{AR}) - \Sm_{\Meas,f}(\rho_{A} \| \sigma_A) \leq \frac{1}{n}$. As the set $\{\rho_{AR} \in \PSD(AR) : \tr_{R} \rho_{AR} = \rho_A\}$ is compact, we can find a converging subsequence $\lim_{i \to \infty} {\rho_{AR}^{(n_i)}} = \rho_{AR}^{\infty}$. By construction, $\tr_{R} \rho_{AR}^{\infty} = \rho_A$ so the data processing inequality (see Remark~\ref{rmk: DPI}) gives $\Sm_{\Meas, f}(\rho^{\infty}_{AR} \| \sigma_{AR}) \geq \Sm_{\Meas,f}(\rho_A \| \sigma_{A})$. As $\Sm_{\Meas,f}$ can be written as the supremum of lower semicontinuous functions, it is also lower semi-continuous. This implies that
\begin{align}
\Sm_{\Meas,f}(\rho_{A} \| \sigma_A) = \lim_{i \to \infty} \Sm_{\Meas,f}(\rho^{(n_i)}_{AR} \| \sigma_{AR}) \geq \Sm_{\Meas,f}(\rho^{\infty}_{AR} \| \sigma_{AR}).
\end{align}
So the infimum is attained at $\rho_{AR}^\infty$.

The proof of Eq.~\eqref{eq: thm uhlmann statement} is similar. Recall the variational expression in Eq.~\eqref{eq:SM variational formula change of variables} with $\psi = (f^*)^{-1}$:
 \begin{align}
    \Sm_{\Meas,f}(\rho \| \sigma) = \sup_{\substack{\gamma \in \HERM\\ \spec(\gamma) \subset \range(f^*)}} \left\{ \tr[\rho (f^*)^{-1}(\gamma)] - \tr[\sigma \gamma] \right\}.
     \label{eq: var expression Qalpha proof uhlmann}
 \end{align}
We now use Sion's minimax theorem~\cite[Corollary 3.3]{Sion1958} to get 
\begin{align}
    \inf_{\substack{\sigma_{AR} \in \PSD(A R) \\ \tr_{R} \sigma_{AR} = \sigma_A}} \Sm_{\Meas, f}(\rho_{AR} \| \sigma_{AR}) =  \sup_{\substack{\gamma_{AR} \in \HERM(AR)\\ \spec(\gamma_{AR}) \subset \range(f^*)}} \inf_{\substack{\sigma_{AR} \in \PSD(A R) \\ \tr_{R} \sigma_{AR} = \sigma_A}} \left\{ \tr[\rho_{AR} (f^*)^{-1}(\gamma_{AR})] - \tr[\sigma_{AR} \gamma_{AR}] \right\}.
\end{align}
In fact, the set $\{\sigma_{AR} \in \PSD(A R) : \tr_{R} \sigma_{AR} = \sigma_A\}$ is compact and convex and $\{\gamma_{AR} \in \HERM(AR): \spec(\gamma_{AR}) \subset \range(f^*)\}$ is convex. In addition the expression $\tr[\rho_{AR} (f^*)^{-1}(\gamma_{AR})] - \tr[\sigma_{AR} \gamma_{AR}]$ is concave in {$\gamma_{AR}$} and linear in {$\sigma_{AR}$}. Now, by semidefinite programming duality, we have
\begin{align}\label{eq: duality support function}
    \sup_{\substack{\sigma_{AR} \in \PSD(A R) \\ \tr_{R} \sigma_{AR} = \sigma_A}} \tr[\sigma_{AR} \gamma_{AR}] = \inf_{\substack{\Lambda_A \in \HERM(A) \\ \Lambda_A \otimes I_{R} \geq \gamma_{AR}}} \tr[\sigma_A \Lambda_A].
\end{align}
In fact it is clear that the supremum is finite as $0 \leq \sigma_{AR} \leq \tr[\sigma_{A}] I_{AR}$ and $\Lambda_{A} = (\| \gamma_{AR} \|_{\infty} + 1) I_{A}$ is strictly feasible for the dual. As a result,
\begin{align}
    \inf_{\substack{\sigma_{AR} \in \PSD(A R) \\ \tr_{R} \sigma_{AR} = \sigma_A}} \Sm_{\Meas, f}(\rho_{AR} \| \sigma_{AR}) =  \sup_{\substack{\gamma_{AR} \in \HERM(AR)\\ \spec(\gamma_{AR}) \subset \range(f^*)}} \sup_{\substack{\Lambda_A \in \HERM(A) \\ \Lambda_A \otimes I_{R} \geq \gamma_{AR}}} \left\{\tr[\rho_{AR} (f^*)^{-1}(\gamma_{AR})] - \tr(\sigma_A \Lambda_A)] \right\}.
\end{align}
Note that we can in particular choose $\gamma_{AR} = \Lambda_A \otimes I_{R}$ and get the easy direction:
\begin{align}
     \inf_{\substack{\sigma_{AR} \in \PSD(A R) \\ \tr_{R} \sigma_{AR} = \sigma_A}} \Sm_{\Meas, f}(\rho_{AR} \| \sigma_{AR})
    &\geq \sup_{\substack{\Lambda_A \in \HERM(A) \\ \spec(\Lambda_{A}) \subset \range(f^*)}} \left\{\tr[\rho_A (f^*)^{-1}(\Lambda_A)] - \tr[\sigma_{A} \Lambda_{A}] \right\} \\
    &= {\Sm_{\Meas,f}(\rho_A \| \sigma_A)}.
\end{align}
For the other direction, as $(f^*)^{-1}$ is operator monotone we have for any feasible $(\gamma, \Lambda)$, it holds that $(f^*)^{-1}(\gamma_{AR}) \leq (f^*)^{-1}(\Lambda_A)\ox I_R$. In addition, as $\sup \range(f^*) = +\infty$ and $\spec(\gamma_{AR}) \subset \range(f^*)$, we have that $\spec(\Lambda_A) \subset \range(f^*)$. In fact, $f^*$ is continuous on its domain since it is a univariate lower semicontinuous convex function, and thus $\range(f^*)$ is an interval, which by our additional assumption, is of the form $[a,+\infty)$ or $(a,+\infty)$. This then implies
\begin{align}
    \inf_{\substack{\sigma_{AR} \in \PSD(A R) \\ \tr_{R} \sigma_{AR} = \sigma_A}} \Sm_{\Meas, f}(\rho_{AR} \| \sigma_{AR})
    &\leq \sup_{\substack{\Lambda_A \in \HERM(A) \\ \spec(\Lambda_{A}) \subset \range(f^*)}} \left\{\tr[\rho_A (f^*)^{-1}(\Lambda_A)] - \tr[\sigma_{A} \Lambda_{A}] \right\} \\
    &= \Sm_{\Meas,f}(\rho_{A} \| \sigma_{A}).
\end{align}
This establishes~\eqref{eq: thm uhlmann statement}. The attainment of the infimum follows the same proof as before.
\end{proof}

As an immediate corollary, we obtain the following result of \Renyi divergences.
\begin{shaded}
\begin{corollary}[Uhlmann's theorem for measured \Renyi divergences]\label{thm: DM Uhlmann}
Let $\rho_{A} \in \density(A)$, $\sigma_{AR} \in \PSD(AR)$,  and $\alpha \in [0,\frac{1}{2}]$, then 
 \begin{align}
 \label{eq: uhlmann statement 1}
     \inf_{\substack{\rho_{AR} \in \PSD(AR) \\ \tr_{R} \rho_{AR} = \rho_A}} D_{\Meas, \alpha}(\rho_{AR} \| \sigma_{AR}) = D_{\Meas,\alpha}(\rho_{A} \| \sigma_A).
 \end{align} 
 Similarly, let $\rho_{AR} \in \density(AR), \sigma_{A} \in \PSD(A)$ and $\alpha \in [\frac{1}{2},+\infty]$~\footnotemark, then 
 \begin{align}
 \label{eq: uhlmann statement}
     \inf_{\substack{\sigma_{AR} \in \PSD(AR) \\ \tr_{R} \sigma_{AR} = \sigma_A}} D_{\Meas, \alpha}(\rho_{AR} \| \sigma_{AR}) = D_{\Meas,\alpha}(\rho_{A} \| \sigma_A).
 \end{align}
 In addition, the infimum in both equations is attained.
 \end{corollary}
 \end{shaded}
 \footnotetext{{When $\alpha = 1$, $D_{\Meas,\alpha}$ is interpreted as $D_\Meas$. When $\alpha = +\infty$, $D_{\Meas,\alpha}$ is interpreted as $D_{\max}$.}}

 \begin{remark}
     Note that for $\alpha \geq 1$, it is easy to see that Uhlmann's property in the first argument (i.e.,~\eqref{eq: uhlmann statement 1}) cannot hold in general. 
     In fact, we can choose $\rho_{A} = \ket{0} \bra{0}$ and {the Bell state} $\sigma_{AR} = \ket{\Phi} \bra{\Phi}$ where $\ket{\Phi} = \frac{1}{\sqrt{2}}(\ket{00} + \ket{11})$, then $D_{\Meas, \alpha}(\rho_{A} \| \sigma_{A}) = \log 2$ but $D_{\Meas,\alpha}(\rho_{AR} \| \sigma_{AR}) = +\infty$ for any extension $\rho_{AR}$ of $\rho_{A}$. 
 \end{remark}
\begin{proof}
    For $\alpha \in [0,\frac{1}{2}]$, as previously mentioned, we have $f_{\alpha}^*(z) = (1-\alpha) (\frac{-z}{\alpha})^{-\frac{\alpha}{1-\alpha}}$ and $\dom(f^*) = \RR_{<0}$. In this case, $f^*_{\alpha}$ is operator convex, operator monotone and the domain is unbounded from below.

    For $\alpha \in [\frac{1}{2},1)$, we also have $f_{\alpha}^*(z) = (1-\alpha) (\frac{-z}{\alpha})^{-\frac{\alpha}{1-\alpha}}$ and $\dom(f_{\alpha}^*) = \RR_{<0}$. But in this case, we check the conditions for~\eqref{eq: thm uhlmann statement}: we have $f_{\alpha}^*$ is one-to-one with $\range(f_{\alpha}^*) = \RR_{>0}$ and $(f_{\alpha}^*)^{-1}(\lambda) = -\alpha (\frac{\lambda}{1-\alpha})^{-\frac{1-\alpha}{\alpha}}$ which is operator concave, operator monotone and we have $\sup \range(f_{\alpha}^*) = +\infty$.

    For $\alpha > 1$, we have $f^*_{\alpha}(z) = (\alpha - 1) \left(\frac{z}{\alpha}\right)^{\frac{\alpha}{\alpha-1}}$ with $\dom(f^*) = \range(f^*) = \RR_{>0}$. Thus, $(f_{\alpha}^*)^{-1}(\lambda) =\alpha (\frac{\lambda}{\alpha-1})^{\frac{\alpha-1}{\alpha}}$ is operator concave, operator monotone and $\sup \range(f_{\alpha}^*) = +\infty$.

    For $\alpha=1$, recall that $f_1(t)=t\log t$, and thus we have $f_1^*(z) = e^{z-1}$ with $\dom(f_1^*) = \RR$ and $\range(f_1^*) = \RR_{>0}$. Thus $(f_1^*)^{-1}(\lambda) = 1+\log \lambda$ is operator concave, operator monotone and $\sup \range(f_1^*) = +\infty$.

    For $\alpha = +\infty$, we take supremum over $\alpha \in [1,+\infty)$. Note that $D_{\Meas,\alpha}(\rho\|\sigma)$ is lower semicontinous in $(\rho,\sigma)$ for every $\alpha \in [1,+\infty)${~\cite[Proposition 18]{mosonyi2023continuitypropertiesquantumrenyi}} and monotonic increasing in $\alpha$ for every $(\rho,\sigma)${~\cite[Proposition 22]{mosonyi2023continuitypropertiesquantumrenyi}}. So we can use the minimax theorem in~\cite[Corollary A.2]{mosonyi2011quantum} to exchange the supremum and infimum. Finally, noting that $\sup_{\alpha \geq 1} D_{\Meas,\alpha}(\rho\|\sigma) = D_{\max}(\rho\|\sigma)$, we have the desired result. {As $D_{\max}(\rho\|\sigma)$ is lower semicontinous in $(\rho,\sigma)$ \cite[Proposition 18]{mosonyi2023continuitypropertiesquantumrenyi}, the infimum is attained on a compact set~\cite[Theorem 7.3.1]{kurdila2005convex}.}
\end{proof}
 
The case $\alpha = \frac{1}{2}$ corresponds to the fidelity as $D_{\Meas,\frac{1}{2}} = -2 \log F$, where $F$ is the fidelity; see~\cite{fuchs1996distinguishability,mosonyi2015quantum}. Thus, Uhlmann's theorem is the special case $\alpha = \frac{1}{2}$ of this theorem. The case $\alpha = +\infty$ was shown in~\cite{metger2024generalised}. 
Moreover, the regularized Uhlmann's theorem from~\cite{metger2024generalised} follows easily from this theorem. Note that the result in~\cite{metger2024generalised} is shown for $\alpha > 1$, whereas the following result applies more generally.

Let $D_{\Sand,\alpha}(\rho\|\sigma) := \frac{1}{\alpha-1}\log\tr\left[\sigma^{\frac{1-\alpha}{2\alpha}}\rho\sigma^{\frac{1-\alpha}{2\alpha}}\right]^\alpha$
be the sandwiched \Renyi divergence~\cite{muller2013quantum,wilde2014strong}.

\begin{shaded}
\begin{corollary}\label{coro: regularized Uhlmann}
Let $\rho_{AR} \in \density(AR), \sigma_{A} \in \PSD(A)$ and $\alpha \in [\frac{1}{2}, +\infty]$, then 
\begin{align}
\label{eq: regularized uhlmann statement}
\lim_{n \to \infty} \frac{1}{n} \inf_{\substack{\sigma_{A^nR^n} \in \PSD(A^n R^n)\\ \tr_{R^n} \sigma_{A^nR^n} = \sigma_A^{\otimes n}}} D_{\Sand, \alpha}(\rho_{AR}^{\otimes n} \| \sigma_{A^nR^n}) = D_{\Sand,\alpha}(\rho_{A} \| \sigma_A).
\end{align} 
\end{corollary}
\end{shaded}
\begin{proof}
Applying Corollary~\ref{thm: DM Uhlmann} to the states $\rho_{AR}^{\otimes n}$ and $\sigma_{A}^{\otimes n}$, we get 
 \begin{align}\label{eq: sand uhlmann tmp1}
     \inf_{\substack{\sigma_{A^nR^n} \in \PSD(A^n R^n) \\ \tr_{R^n} \sigma_{A^nR^n} = \sigma_A^{\otimes n}}} D_{\Meas, \alpha}(\rho_{AR}^{\otimes n} \| \sigma_{A^nR^n}) = D_{\Meas,\alpha}(\rho_{A}^{\otimes n} \| \sigma_A^{\otimes n}).
 \end{align} 
By the asymptotic equivalence of the measured \Renyi divergence and the sandwiched \Renyi divergence~\cite[Lemma 28]{fang2024generalized}, applied with $\sB_n = \{\sigma_{A^nR^n} \in \PSD(A^n R^n): \tr_{R^n} \sigma_{A^nR^n} = \sigma_A^{\otimes n}\}$ (which is convex, compact, and permutation invariant), we obtain {
 \begin{align}\label{eq: sand uhlmann tmp2}
     \lim_{n \to \infty} \frac{1}{n} \inf_{\sigma_{A^n R^n} \in \sB_n}D_{\Sand, \alpha}(\rho_{AR}^{\otimes n} \| \sigma_{A^n R^n}) &= \lim_{n \to \infty} \frac{1}{n} \inf_{\sigma_{A^n R^n}\in \sB_n} D_{\Meas, \alpha}(\rho_{AR}^{\otimes n} \|\sigma_{A^n R^n}).
 \end{align}      
Combining Eq.~\eqref{eq: sand uhlmann tmp1}, Eq.~\eqref{eq: sand uhlmann tmp2} and the fact that $\lim_{n\to \infty}\frac{1}{n} D_{\Meas,\alpha}(\rho_A^{\ox n}\|\sigma_A^{\ox n}) = D_{\Sand,\alpha}(\rho_A\|\sigma_A)$ \cite[Corollary 3.8]{mosonyi2015quantum}, we have the asserted result.}
\end{proof}

\begin{remark}
\label{rem: Bn hypothesis testing}
Note that the set $\sB_n$ defined in the proof above satisfies all assumptions (A.1)-(A.4) in~\cite[Assumption 25]{fang2024generalized}. In fact, (A.1) and (A.2) are clear and $\tr_{R_1 R_2} (\sigma^{(1)}_{A_1 R_1} \otimes \sigma^{(2)}_{A_2 R_2}) = \sigma_{A_1}^{(1)} \otimes \sigma_{A_2}^{(2)}$ so (A.3) is satisfied. The condition (A.4) states that the support function is submultiplicative, i.e., $h_{\sB_{m+k}}(X_{A^mR^m} \otimes X_{A^kR^k}) \leq h_{\sB_{m}}(X_{A^mR^m}) h_{\sB_k}(X_{A^kR^k})$. For that, recall that we have obtained a dual formulation for the support function of the sets $\sB_n$ in~\eqref{eq: duality support function}. We use this same dual formulation to establish the desired submultiplicativity:
\begin{align*}
    h_{\sB_{m+k}}(X_{A^mR^m} \otimes X_{A^kR^k}) 
    &=      \sup_{\substack{\sigma_{A^{m+k} R^{m+k}} \in \PSD(A^{m+k} R^{m+k}) \\ \tr_{R^{m+k}} \sigma_{A^{m+k} R^{m+k}} = \sigma_A^{\otimes m+k}}} \tr(\sigma_{A^{m+k} R^{m+k}} X_{A^mR^m} \otimes X_{A^kR^k}) \\
    &= \inf_{\substack{\Lambda_{A^{m+k}} \in \HERM(A^{m+k}) \\ \Lambda_{A^{m+k}} \otimes I_{R^{m+k}} \geq X_{A^mR^m} \otimes X_{A^mR^m}}} \tr(\sigma_A^{\otimes m + k} \Lambda_{A^{m+k}}) \\
    &\leq \inf_{\substack{\Lambda_{A^{m}} \in \HERM(A^{m}),\Lambda_{A^{k}} \in \HERM(A^{k})  \\ \Lambda_{A^{m}} \otimes \Lambda_{A^k} \otimes I_{R^{m+k}} \geq X_{A^mR^m} \otimes X_{A^kR^k}}} \tr(\sigma_A^{\otimes m + k} \Lambda_{A^{m}} \otimes \Lambda_{A^k}) \\
    &\leq \inf_{\substack{\Lambda_{A^{m}} \in \HERM(A^{m}),\Lambda_{A^{k}} \in \HERM(A^{k})  \\ \Lambda_{A^{m}} \otimes I_{R^m} \geq X_{A^mR^m} \\
    \Lambda_{A^k} \otimes I_{R^{k}} \geq X_{A^kR^k} }} \tr(\sigma_A^{\otimes m} \Lambda_{A^{m}}) \tr(\sigma_A^{\otimes k} \Lambda_{A^k}) \\ 
    &= h_{\sB_m}(X_{A^mR^m}) h_{\sB_k}(X_{A^kR^k}).
\end{align*}
Thus, the results of~\cite{fang2024generalized} characterize the asymmetric hypothesis testing problem between $\{\rho_{AR}^{\otimes n}\}$ and $\sB_n$. 
 A related hypothesis testing scenario was considered in~\cite{hayashi2016correlation}. The authors of~\cite{hayashi2016correlation} studied hypothesis testing between $\{\rho_{AR}^{\ox n}\}$ and $\{\rho_A^{\ox n} \ox \sigma_{R^n}: \sigma \in \density(R^n)\}$  
 and demonstrated that the \Renyi mutual information  
 is the strong converse exponent in this context. 
 Similarly, we expect Corollary~\ref{coro: regularized Uhlmann} can be used to establish 
 the strong converse exponent for hypothesis testing between $\rho_{AR}^{\otimes n}$ and $\sB_n$. 
\end{remark}

\section{Duality with the quantum relative entropy}

\label{sec: duality}

This section is specific to the measured {relative entropy} (i.e., $f(t) = t \log t$): we utilize the variational expression for $D_{\Meas}$ to derive a duality relation with the Umegaki quantum relative entropy which is defined by $D(\rho \| \sigma) = \tr(\rho \log \rho) - \tr(\rho \log \sigma)$ when the support of $\rho$ is included in the support of $\sigma$ and $+\infty$ otherwise.
For a quantum divergence $\DD$ and a set $\cvxset$ of positive operators, we write $\DD(\rho\|\cvxset):= \inf_{\sigma \in \cvxset} \DD(\rho\|\sigma)$.

\begin{shaded}
\begin{proposition}\label{thm: duality sum zero}
Let $\rho \in \density$ and $\cvxset \subset \PSD$ be {an non-empty} compact convex set, {such that $\supp(\rho) \subset \supp(\omega)$ for some $\omega \in \cvxset$. Then it holds that}
\begin{align}
    D(\rho\|\cvxset) + D_{\Meas}(\rho\|\cvxset_{\pl \pl}^{\circ}) = D(\rho\|I),
\end{align}    
where $\cvxset^\circ:= \{X: \tr[XY] \leq 1,  \forall\, Y\in \cvxset\}$ is the polar set and $\cvxset^\circ_{\pl \pl}:= \cvxset^\circ \cap \PD$.
\end{proposition}
\end{shaded}
\begin{proof}
By the variational expression for the measured {relative entropy}~\cite{petzquantum,Berta2017} also presented in~\eqref{eq:kl measured}, we have
\begin{align}
D_{\Meas}(\rho\|\sigma) & = \sup_{\omega \in \PD} \ \tr[\rho \log \omega] + 1  - \tr[\sigma \omega]\\
& = \sup_{\omega \in \PD} \ -D(\rho\|\omega) - \tr[\sigma \omega] + 1 + D(\rho\|I)\\
& = g^*_{\rho}(-\sigma)  + 1 + D(\rho\|I),\label{eq: DM and Dstar}
\end{align}
where the second line follows from the definition of Umegaki relative entropy and the third line follows by denoting $g_\rho^*$ as the Fenchel conjugate of $g_{\rho} : \omega \mapsto D(\rho \| \omega)$,  with $\dom(g_{\rho}) = \{\omega \geq 0 : \rho \ll \omega\}$, i.e., for $\tau \in \HERM$, $g^*_{\rho}(\tau) = \sup_{\omega > 0} \tr[\omega \tau] - D(\rho \| \omega)$.\footnote{Note that $g_{\rho}^*(\tau) := \sup_{\omega \in \dom(g_{\rho})} \tr[\omega \tau] - D(\rho \| \omega) = \sup_{\omega > 0} \tr[\omega \tau] - D(\rho \| \omega)$ since for any $\omega \geq 0$, $D(\rho \| \omega) = \lim_{\epsilon\downarrow0} D(\rho \| \omega + \epsilon I)$} {Let us show that $\dom(g_{\rho}^*) = \{ \tau \leq 0 : \supp(\rho) \subset \supp(\tau)\}$. In fact, we can write 
\begin{align}
    g_{\rho}^*(\tau) 
    &= \sup_{\{\lambda_i, \Pi_i\}_i} \sum_{i} \lambda_i \tr[\Pi_i \tau] + \log \lambda_i \tr[\rho \Pi_i] - \tr[\rho \log \rho] \\
    &= -\tr[\rho \log \rho] + \sup_{\{\Pi_i\}_i} \sum_{i} \sup_{\{\lambda_i\}_i} \lambda_i \tr[\Pi_i \tau] + \log \lambda_i \tr[\rho \Pi_i] .
\end{align}
where the supremum is over orthogonal projectors $\Pi_i$ with $\sum_{i} \Pi_i = I$ and $\lambda_i > 0$. This quantity is finite if and only if for all $i$, we have $\tr[\Pi_i \tau] < 0$ or if $\tr[\Pi_i \tau] = 0$, then $\tr[\rho \Pi_i] = 0$. 
}
As $g_{\rho}$ 
is lower semicontinuous, the biduality relation of Fenchel conjugate~\cite[Theorem 12.2]{rockafellar} gives
\begin{align}
    D(\rho\|\sigma) = g_{\rho}(\sigma) & = \sup_{\omega \in \dom(g_{\rho}^*)} \tr [\sigma \omega] - g_{\rho}^*(\omega)\\
    & = \sup_{\omega < 0} \tr [\sigma \omega] - g_{\rho}^*(\omega)\\
    & = \sup_{\omega < 0} \tr [\sigma \omega] - D_{\Meas}(\rho\|-\omega) + D(\rho\|I) + 1\label{eq:DrhosigmaFenchel3} \\
    & = \sup_{\omega > 0} - \tr [\sigma \omega] - D_{\Meas}(\rho\|\omega) + D(\rho\|I) + 1.\label{eq:DrhosigmaFenchel4}
\end{align}
In the second line we used the fact that for any $\omega \in \dom(g_{\rho}^*)$ {and any $\epsilon > 0$, we have $(1-\epsilon) \omega + \epsilon (-I) < 0$ and}
\[
g_{\rho}^*(\omega) = \lim_{\epsilon\downarrow0} g_{\rho}^*((1-\epsilon) \omega + \epsilon (-I))
\]
which follows by applying \cite[Theorem 7.5]{rockafellar}. Note that the theorem applies because $-I$ is in the interior of $\dom(g_{\rho}^*)$ and $g_{\rho}^*$ is lower semicontinuous.
The third line \eqref{eq:DrhosigmaFenchel3} follows from Eq.~\eqref{eq: DM and Dstar}, and the last line follows by replacing $\omega$ with $-\omega$.
Optimizing over $\sigma \in \cvxset$, we have
\begin{align}
    D(\rho\|\cvxset)  = \inf_{\sigma \in \cvxset} D(\rho\|\sigma) = \inf_{\sigma \in \cvxset} \sup_{\omega \in \PD} - \tr[\sigma \omega] - D_{\Meas}(\rho\|\omega) + 1 + D(\rho\|I).
\end{align}
Since $D_{\Meas}(\rho\|\omega)$ is convex in $\omega$, the above objective function is concave in $\omega$ and {convex (actually linear)} in $\sigma$. Moreover, since $\cvxset$ is a compact convex set by assumption, we can use Sion's minimax theorem~\cite[Corollary 3.3]{Sion1958} to exchange the infimum and supremum. 
This gives
\begin{align}
   D(\rho\|\cvxset) = \sup_{\omega \in \PD} \inf_{\sigma \in \cvxset}  - \tr[\sigma \omega] - D_{\Meas}(\rho\|\omega) + 1 + D(\rho\|I).
\end{align}
Let $h_{\cvxset}(\omega) := \sup_{\sigma \in \cvxset} \tr [\sigma \omega]$ be the support function of $\cvxset$. We get
\begin{align}
   D(\rho\|\cvxset)  
    &  = \sup_{\omega \in \PD}  - h_{\cvxset}(\omega) - D_{\Meas}(\rho\|\omega) + 1 + D(\rho\|I)\\
  &  \stackrel{(a)}{=} \sup_{\substack{\omega \in \PD\\h_{\cvxset}(\omega)= 1}}  - D_{\Meas}(\rho\|\omega) + D(\rho\|I)\\
  &  \stackrel{(b)}{=} \sup_{\substack{\omega \in \PD\\h_{\cvxset}(\omega) \leq 1}}    - D_{\Meas}(\rho\|\omega) + D(\rho\|I)\\
  & \stackrel{(c)}{=} - D_{\Meas}(\rho \| \cvxset_{\pl \pl}^\circ) + D(\rho\|I),
  \end{align} 
To see why $(a)$ holds, let $\omega$ be {any feasible} solution in the first line and define $\widetilde \omega = \omega / h_{\cvxset}(\omega)$. Then $h_{\cvxset}(\widetilde \omega) = 1$ and we will see that $\widetilde \omega$ achieves an objective value no smaller than $\omega$ as
\begin{align}
- D_{\Meas}(\rho\|\widetilde \omega) = -D_{\Meas}(\rho\|\omega) - \log h_\cvxset(\omega) \geq - D_{\Meas}(\rho\|\omega) -  h_\cvxset(\omega) + 1
\end{align} 
where the inequality follows from the fact that $\log x \leq x-1$. The step $(b)$ follows by the same argument. Suppose $\omega$ is a solution in the third line, let $\widetilde \omega = \omega / h_{\cvxset}(\omega)$ and we can check that $\widetilde \omega$ gives an objective value no smaller than the second line. Finally, we have the step $(c)$ by the fact that $h_{\cvxset}(\omega) \leq 1$ if and only if $\omega \in \cvxset^{\circ}$ and the notation that $\cvxset_{\pl \pl}^\circ = \cvxset^\circ \cap \PD$. This completes the proof.
\end{proof}

\begin{remark} {
By the bipolar theorem $(\cvxset^{\circ})^{\circ} = \overline{\operatorname{conv}(\cvxset \cup \{0\})}$~\cite[Exercise 1.15]{aubrun2017alice} together with the dominance property
$D_{\Meas}(\rho\|\lambda \sigma) \ge D_{\Meas}(\rho\|\sigma)$ for $\lambda\in(0,1]$, one obtains a similar duality
$
D(\rho\|\cvxset^{\circ}_{\pl\pl}) + D_{\Meas}(\rho\|\cvxset) = D(\rho\|I).
$
Such duality relations allow one to transfer additivity properties between the two terms. For example, if the polar set $(\cvxset_n)^{\circ}$ is closed under tensor products—a technical assumption in many recent applications (e.g.,~\cite{fang2024generalized,fang2025efficient,fang2025adversarial,arqand2025marginal})—then $D(\rho^{\otimes n}\|(\cvxset_n)^{\circ}_{++})$ is clearly subadditive by definition, which in turn implies superadditivity of $D_{\Meas}(\rho^{\otimes n}\|\cvxset_n)$. This superadditivity underpins several recent results, as cited above. The duality relation thus provides a more intuitive explanation for why measured relative entropy exhibits superadditivity when the polar sets satisfy this property. The same reasoning applies when swapping the roles of $D$ and $D_{\Meas}$. These insights are in the same spirit as recent progress on the Frenkel representation of quantum relative entropy~\cite{frenkel2023integral}, which offers a direct route to proving the data-processing inequality—often challenging from the original definition, but straightforward from this new perspective. }
\end{remark}

Since the measured relative entropy coincides with the quantum relative entropy in the classical case, we can replace $D_{\Meas}$ with $D$, obtaining an equality expressed solely in terms of the relative entropy. In the quantum case, if we consider the regularized divergence $\DD^\infty(\rho\|\cvxset):= \lim_{n\to \infty} {\frac{1}{n}} \DD(\rho^{\ox n}\|\cvxset_n)$, we get a self-duality for the quantum relative entropy as follows.

\begin{shaded}
\begin{corollary}
Let $\rho \in \density$. Let $\{\sC_n\}_{n\in\NN}$ be a sequences of sets satisfying $\sC_n \subset \PSD(\cH^{\ox n})$ and $D_{\max}(\rho^{\ox n}\|\sC_n) \leq cn$, for all $n \in \NN$ and a constant $c \in \RR_{\pl}$. In addition, each $\cvxset_n$ is convex, compact and permutation-invariant, and $\cvxset_m \otimes \cvxset_k \subset \cvxset_{m+k}$ for any $m,k\in \NN$. Then it holds that
\begin{align}
 D^\reg(\rho\|\cvxset) + D^\reg(\rho\|\cvxset_{\pl \pl}^{\circ}) = D(\rho\|I).
\end{align} 
\end{corollary}
\end{shaded}
\begin{proof}
By Proposition~\ref{thm: duality sum zero}, we know that
\begin{align}
    \frac{1}{n} D(\rho^{\ox n}\|\cvxset_n) + \frac{1}{n} D_{\Meas}(\rho^{\ox n}\|(\cvxset_n)_{\pl \pl}^{\circ}) = D(\rho\|I).
\end{align}
{Since the first term on the left-hand side is subadditive (see~\cite[Lemma 26]{fang2024generalized}), its limit exists when $n$ goes to infinity by Fekete's lemma.} By the equality relation, the limit for the second term on the left-hand side also exists. Therefore, we have 
\begin{align}
    \left[ \lim_{n\to \infty} \frac{1}{n} D(\rho^{\ox n}\|\cvxset_n) \right] + \left[\lim_{n\to \infty} \frac{1}{n} D_{\Meas}(\rho^{\ox n}\|(\cvxset_n)_{\pl \pl}^{\circ})\right] = D(\rho\|I).
\end{align}
Finally, we complete the proof by noting that the second term on the left-hand side is equal to the one with quantum relative entropy (see~\cite[Lemma 28]{fang2024generalized}).
\end{proof}

\section{Discussions}
\label{sec: discussions}

{In this work, we have established Uhlmann's theorem for general measured $f$-divergences, which includes the measured \Renyi divergences as special cases and generalizes the well-known Uhlmann's theorem for the fidelity. The significance and potential applications are as follows:
\begin{itemize}
\item (\textbf{Fundamental importance:}) Uhlmann's theorem is a cornerstone in quantum information theory. It has applications in various subareas, ranging from quantum Shannon theory~\cite{watrous2018theory} and quantum computing~\cite{bostanci2025unitarycomplexityuhlmanntransformation} to quantum gravity~\cite{Kirk20}. A concrete example for the latter is the construction of the holographic dual of the bulk symplectic form in an entanglement wedge. The construction relies on evaluating the fidelity via a particular replica trick and uses Uhlmann's theorem as a key ingredient. This naturally prompts the question (see Footnote 4 of~\cite{Kirk20}) of whether the construction extends to \Renyi divergences and, if so, whether the resulting conclusions remain consistent. 
Moreover, having an Uhlmann's theorem for a quantum divergence is of fundamental interest, as it serves as a criterion for classifying divergences into different types. Note that Uhlmann's theorem cannot hold for any ``sufficient'' quantum divergence (in the sense that saturation of the data processing inequality $\mathbb{D}(\Phi(\rho)\|\Phi(\sigma)) = \mathbb{D}(\rho\|\sigma)$ implies the reversibility of $\Phi$ on $\{\rho, \sigma\}$). Since most commonly used quantum \Renyi divergences, including the Petz and sandwiched \Renyi divergences, are sufficient in this sense, they cannot satisfy Uhlmann's theorem (except for degenerate cases). This fundamentally distinguishes measured $f$-divergences from other quantum divergences and highlights their unique mathematical structure.
\item (\textbf{Existing application:}) In the case of $\alpha > 1$, the Uhlmann's theorem for regularized sandwiched \Renyi divergence proved in~\cite{metger2022generalised} already plays a crucial role in deriving the generalized entropy accumulation theorem, which is then used in quantum cryptographic tasks such as blind randomness expansion. Our work establishes Uhlmann's theorem directly for the measured \Renyi divergence in the single-letter case, which provides a conceptually clearer and technically simpler proof of the chain rule in~\cite{metger2022generalised}. Our Uhlmann theorem for measured divergences can potentially be used to obtain improved chain rules, e.g. improving the measured divergence term in~\cite[Corollary 4]{berta2022chain} by taking partial traces when the channel $\mathcal{F}$ satisfies a non-signalling property. We leave this for future work.
\item (\textbf{Potentially new application:}) The Uhlmann's theorem for measured divergence provides a potential approach to computing the amortized relative entropy for quantum channels. Let $\cN,\cM$ be two quantum channels; their amortized relative entropy is defined as  
\begin{align}
    D^A(\cN\|\cM) := \sup_{\rho_{RA},\sigma_{RA}} D(\cN_{A\to B}(\rho_{RA})\|\cM_{A\to B}(\sigma_{RA})) - D(\rho_{RA}\|\sigma_{RA}),
\end{align}
where the supremum is over all bipartite states $\rho_{RA},\sigma_{RA}$ with $R$ being a reference system of arbitrary dimension. This quantity serves as the ultimate limit for the asymptotic quantum channel discrimination problem, but its computability remains open~\cite{WW2019,fang2020chain}. The main challenge for this is the potentially unbounded dimension of the reference system $R$ in the optimization. However, our Uhlmann's theorem implies that the dimension in an analogous formula for measured relative entropy can be restricted to the same size as system $A$. Specifically, if we define the amortized measured relative entropy for channels by
\begin{align}
    D_{\Meas}^A(\cN\|\cM) := \sup_{\rho_{RA},\sigma_{RA}} D_{\Meas}(\cN_{A\to B}(\rho_{RA})\|\cM_{A\to B}(\sigma_{RA})) - D_{\Meas}(\rho_{RA}\|\sigma_{RA}),
\end{align} 
and consider a purification of $\rho_{RA}$ as $\ket{\phi}_{ERA}$, then by Uhlmann's theorem, there exists an extension $\sigma_{ERA}$ for $\sigma_{RA}$ such that $D_{\Meas}(\phi_{ERA}\|\sigma_{ERA}) = D_{\Meas}(\rho_{RA}\|\sigma_{RA})$. Moreover, by the data-processing inequality, we have
\begin{align}
    D_{\Meas}(\cN_{A\to B}(\rho_{RA})\|\cM_{A\to B}(\sigma_{RA})) \leq D_{\Meas}(\cN_{A\to B}(\phi_{ERA})\|\cM_{A\to B}(\sigma_{ERA})).
\end{align}
This implies that
\begin{align}
    D_{\Meas}^A(\cN\|\cM) = \sup_{\phi_{ERA},\sigma_{ERA}} D_{\Meas}(\cN_{A\to B}(\phi_{ERA})\|\cM_{A\to B}(\sigma_{ERA})) - D_{\Meas}(\phi_{ERA}\|\sigma_{ERA}).
\end{align}
That is, the optimization can be restricted to pure states $\phi_{ERA}$.
Rewriting the systems $ER$ as $R$, we get 
\begin{align}
    D_{\Meas}^A(\cN\|\cM) = \sup_{\phi_{RA},\sigma_{RA}} D_{\Meas}(\cN_{A\to B}(\phi_{RA})\|\cM_{A\to B}(\sigma_{RA})) - D_{\Meas}(\phi_{RA}\|\sigma_{RA}).
\end{align}
Suppose $|R|\geq |A|$ and let $A'$ be isomorphic to $A$. For any two purifications $\phi_{RA}$ and $\phi_{A'A}$ of $\phi_A$, there exists an isometry $V_{R\to A'}$ such that $\ket{\phi}_{A'A} = V_{R\to A'} \ket{\phi}_{RA}$. By isometry invariance of the measured relative entropy, we get 
\begin{align}
    D_{\Meas}^A(\cN\|\cM) = \sup_{\phi_{A'A},\sigma_{A'A}} D_{\Meas}(\cN_{A\to B}(\phi_{A'A})\|\cM_{A\to B}(\sigma_{A'A})) - D_{\Meas}(\phi_{A'A}\|\sigma_{A'A}),
\end{align} 
with $|A'|=|A|$. Let $D(\cN\|\cM):=\sup_{\rho_{RA}} D(\cN_{A\to B}(\rho_{RA})\|\cM_{A\to B}(\rho_{RA}))$ and its regularization $D^\infty(\cN\|\cM):= \lim_{n\to \infty} \frac{1}{n} D(\cN^{\ox n}\|\cM^{\ox n})$. We can have similar definitions for $D_\Meas$. Then we have the relations that
\begin{align}
\frac{1}{n}D_{\Meas}(\cN^{\ox n}\|\cM^{\ox n}) \leq D^\reg(\cN\|\cM) = D_{\Meas}^\reg(\cN\|\cM) \leq \frac{1}{n}D_{\Meas}^A(\cN^{\ox n}\|\cM^{\ox n}),
\end{align}
where the first inequality follows as $D_\Meas(\cdot\|\cdot) \leq D(\cdot\|\cdot)$, the equality follows by the asymptotic equivalence of measured relative entropy and quantum relative entropy for permutation-invariant states (e.g.~\cite[Lemma 16, 17]{fang2024generalized}), and the last inequality follows from the fact that $D_\Meas(\cN\|\cM) \leq D^A_{\Meas}(\cN\|\cM)$ and the subadditivity of $D^A_{\Meas}(\cN\|\cM)$ by definition.
Note that $D^\reg(\cN\|\cM) = D^A(\cN\|\cM)$ as proved in~\cite{fang2020chain}. Therefore, we have
\begin{align}
\frac{1}{n} D_{\Meas}(\cN^{\otimes n}\|\cM^{\otimes n}) \leq D^A(\cN\|\cM) \leq \frac{1}{n} D_{\Meas}^A(\cN^{\otimes n}\|\cM^{\otimes n}).
\end{align}
It is clear that the lower bound $\frac{1}{n} D_{\Meas}(\cN^{\otimes n}\|\cM^{\otimes n})$ converges to $D^\reg(\cN\|\cM)$ and therefore $D^A(\cN\|\cM)$. If we can show that the upper bound converges to $D^A(\cN\|\cM)$, then we would have established a sandwiching of $D^A(\cN\|\cM)$ between two sequences converging to it, and both sequences can be computed in finite time. This would imply the computability of $D^A(\cN\|\cM)$; that is, there exists an algorithm that terminates in finite time within additive error. 
\end{itemize}
}

\vspace{1cm}

\noindent \textbf{Acknowledgements.} 
We would like to thank David Sutter for a comment that helped us clarify Corollary~\ref{coro: regularized Uhlmann} and Remark~\ref{rem: Bn hypothesis testing}. {We also thank anonymous reviewers for their valuable comments and for bringing Hiai's book to our attention.}
K.F. is supported by the National Natural Science Foundation of China (Grant No. 92470113 and 12404569), the Shenzhen Science and Technology Program (Grant No. JCYJ20240813113519025), the Shenzhen Fundamental Research Program (Grant No. JCYJ20241202124023031), the General R\&D Projects of 1+1+1 CUHK-CUHK(SZ)-GDST Joint Collaboration Fund (Grant No. GRDP2025-022), and the University Development Fund (Grant No. UDF01003565). H.F. was partially funded by UK Research and Innovation (UKRI) under the UK government’s Horizon Europe funding guarantee EP/X032051/1.
O.F. acknowledges support by the European Research Council (ERC Grant AlgoQIP, Agreement No. 851716), by the European Union’s Horizon 2020 research and innovation programme under Grant Agreement No 101017733 (VERIqTAS) and by the Agence Nationale de la Recherche under the Plan France 2030 with the reference ANR-22-PETQ-0009.

\bibliographystyle{alpha_abbrv}
\bibliography{Bib}

\end{document}